\documentclass[11pt,a4paper,final]{article}
\usepackage{mathtools}
\usepackage[linesnumbered,ruled,vlined]{algorithm2e}
\usepackage{amsmath,amsthm,amssymb}
\usepackage[T1]{fontenc} % Active encoding for use in math text
\usepackage{float}
\usepackage{cite}
\usepackage{fullpage}
\newlength{\bibitemsep}
\newlength{\bibparskip}
\let\oldthebibliography\thebibliography
\renewcommand\thebibliography[1]{%
  \oldthebibliography{#1}%
  \setlength{\parskip}{\bibitemsep}%
  \setlength{\itemsep}{\bibparskip}%
}
\DeclarePairedDelimiter{\floor}{\lfloor}{\rfloor}
\newtheorem{defn}{Definition}
\newtheorem{theorem}{Theorem}

\title{A McEliece cryptosystem using permutation codes}
\author{Adarsh Srinivasan \and Ayan Mahalanobis}
\newcommand{\Addresses}{{% additional braces for segregating \footnotesize
  \bigskip
  \footnotesize
  
  A.~Srinivasan, \textsc{Department of Computer Science, Rutgers University, USA}\par\nopagebreak
  \textit{E-mail address}: \texttt{adarshsrinivasan256@gmail.com}
\medskip

A.~Mahalanobis (Corresponding author), \textsc{IISER Pune, Pune, INDIA}\par\nopagebreak
  \textit{E-mail address}: \texttt{ayan.mahalanobis@gmail.com}
}}
\begin{document}
\date{}
\maketitle
\begin{abstract}
This paper is an attempt to build a new public-key cryptosystem; similar to the McEliece cryptosystem, using permutation error-correcting codes. We study a public-key cryptosystem built using two permutation error-correcting codes. We show that these cryptosystems are insecure. However, the general framework in these cryptosystems can use any permutation error-correcting code and is interesting. We present an enhanced McEliece cryptosystem which subsumes McEliece cryptosystem based on linear error correcting codes.   
\end{abstract}
\section{Introduction} 
McEliece and Niederreiter cryptosystems are very popular these days. One of the reason behind this popularity is that there are some instances of these cryptosystems that can resist \emph{quantum Fourier sampling attacks}. These makes them secure, even when quantum computers are built. These kind of cryptosystems are called \emph{quantum-secure cryptosystems}.

McEliece and Niederreiter cryptosystems use linear error-correcting codes. Similar to linear error-correcting codes, we can define \textbf{permutation error-correcting codes}. These codes use permutation groups the same way linear codes use vector-spaces. This paper is an attempt to build secure public-key cryptosystems using permutation codes in the same spirit as McEliece and Niederreiter cryptosystems were built using linear error-correcting codes. The cryptosystems that came out of permutation codes is not secure. However, the journey we took to build these cryptosystems is interesting in its own right. Moreover, we were able to \textbf{enhance the original McEliece cryptosystem} by embedding it into a permutation group.

In this paper we ask two questions:
\begin{description}
\item[Q1] Can one build a secure public-key cryptosystem using permutation codes, similar to cryptosystems built using linear error-correcting codes?
\item[Q2] Are there any advantages of using permutation codes compared to linear codes in public-key cryptosystems?
\end{description}
We try to keep this paper as self contained as possible. In the next section, we present a brief overview of permutation codes. Most of it is a review of Bailey's work~\cite{bailey1,bailey2,bailey_thesis} on permutation codes. We have performed some extra computational analysis on his alternative decoding in Section~\ref{analysis}. The section after that presents a public-key cryptosystem using permutation codes whose security depends on the hardness of decoding generic permutation codes. We study this cryptosystem using two permutation codes -- the symmetric group acting on 2-subsets and a class of wreath product groups. We show that both these cryptosystems are insecure. The first one is insecure due to an information set decoding attack and the latter due to an inherent structure in the wreath product. We then present the enhanced McEliece cryptosystem in Section~\ref{enhanced}.
\section{Permutation error-correcting codes}
The use of sets of permutations in coding theory (also referred to as permutation arrays) was studied since the 1970's. Blake et al.~\cite{blake,blake1979coding} were the first to discuss using permutations this way. They had certain applications in mind. One such example was powerline communications. However, permutation codes have never got the level of attention that linear codes did. Bailey~\cite{bailey_thesis}, in his thesis, describes a variety of permutation codes and presents a decoding algorithm which works for arbitrary permutation groups using a combinatorial structure called \emph{uncovering-by-bases}. He was the first to exploit the algebraic structure of groups to come up with decoding algorithms for several families of groups. 

In this section, we explore the use of permutation groups as error-correcting codes. We begin by defining a Hamming metric on permutations, similar to the one defined for vector spaces over finite fields. Let us define a set $\Omega=\{1,2,\ldots,n\}$. For a permutation $g \in S_n$, we define its support $Supp(g)$ and set of fixed points $Fix(g)$ to be those points it moves and those points it fixes respectively. A permutation, which is a bijective function from $\Omega$ to $\Omega$ can be represented in two ways. The first is by listing down all its images. This is called the list form of a permutation. For example, $[2,3,1]$ is a permutation $g$ acting on $\{1,2,3\}$ such that $1 \cdot g=2,2 \cdot g=3,3 \cdot g=1$. The other way, as a product of disjoint cycles. The list form of a permutation is useful in defining a Hamming metric. 
\begin{defn}
For elements $g,h \in S_n$, the Hamming distance between them, $d(g,h)$ is defined to be the size of the set $\{x \in \Omega | x \cdot g \neq x \cdot h\}$.   
\end{defn}
In other words, the Hamming distance between two permutations is the the number of positions on which they differ  when written down in list form. The Hamming distance between two permutations $g$ and $h$ can be proved to be equal to $|Supp(gh^{-1})|=n-|Fix(gh^{-1})|$. 
A \emph{permutation code} is a collection of permutations in $S_n$ endowed with a Hamming metric. For a subset $C$ of $S_n$, the minimum distance of the subset is defined to be $\min_{g,h \in C}d(g,h)$. In this paper, we assume that the collection of permutations form a subgroup $G$ of $S_n$. In that case, the minimum distance is:
\begin{equation}
    m(G)=\min_{g \in G, g \neq 1}|Supp(g)|=n-\max_{g \in G, g \neq 1}|Fix(g)|.
\end{equation}
This is also known as the \emph{minimal degree} of the permutation group. The quantity $r=\floor*{\frac{m(G)-1}{2}}$ is called the \emph{correction capacity} of the code. 

We can now define the distance between a permutation $\sigma$ and a group $G$ to be the distance between $\sigma$ and the permutation in $G$ closest to it.
\begin{equation*}
    d(G,\sigma)=\min_{g \in G}d(g,\sigma)
\end{equation*}
The subgroup distance problem can be defined as follows: Given a set of generators of a permutation group $G \leq S_n$, a permutation $\sigma \in S_n$ and an integer $k > 0$, is $d(\sigma,G) \leq k$? We can also generalize this problem to consider inputs $\sigma$ to be lists of length $n$ with symbols from $\Omega$ which do not necessarily form a permutation.

A nearest neighbor of $\sigma$ in $G$ is defined to be a permutation $g$ in $G$ such that $d(g,\sigma)=d(G,\sigma)$. This permutation exists because $d(g,\sigma)=\min_{g \in G}d(g,\sigma)$. If the distance between $\sigma$ and $G$ is greater than the correction capacity of the group, this nearest neighbor need not be unique. However, if $d(\sigma,G) \leq r$, the nearest neighbor of $\sigma$ in $G$ is unique. 

Hence, $G$ can be used as an error-correcting code in the following manner. The information we want to transmit is encoded as a permutation in $G$ and transmitted as a list through a noisy channel. We assume that when the message was received, $r$ or fewer errors were introduced. Such a channel is called a \emph{Hamming channel}. The received message is a list $w$ of length $n$ with symbols from $\Omega$ (which need not necessarily form a permutation). Because the correction capacity of $G$ is $r$, the nearest neighbor of $w$ in $G$ is uniquely defined to be $g$. The decoder can now solve the subgroup distance problem with the word $w$ and the group $G$ as input to recover the original message $g$. However, while we have proved that decoding is possible in theory, because the \emph{subgroup distance problem is NP-hard}, generic algorithms cannot be used to decode permutation codes efficiently. Rather, we need specific algorithms tailor-made for different families of permutation codes.
\subsection{A decoding algorithm using uncovering-by-bases}
In this section, we review Bailey's work on decoding algorithms for permutation codes. Consider a permutation group $G \leq S_n$ with correction capacity $r$. Bailey's decoding algorithm uses the following combinatorial structure. An uncovering-by-bases (UBB) is a collection of bases for $G$. To recall, a base is a set of points $B = \{b_1,b_2,\ldots,b_m\}\subseteq\Omega$ with the property that the only element of $G$ which stabilises all the points is the identity. Given a base, we can define the following subgroup chain:
\begin{equation}
    G=G^{[0]} \geq G^{[1]} \dots \geq G^{[m]}=\{1\}
\end{equation}
where $G^{[i]}$ is the pointwise stabilizer of $\{b_1,b_2,\dots,b_i\}$ in $G$. We say that a base $B$ is \emph{irredundant} if $G^{[i]} \neq G^{[i+1]}$ for all $i$. A generating set $S$ of $G$ is called a strong generating set if $\left<S \cap G^{[i]} \right>=G^{[i]}$ for $1 \leq i \leq m$. Here $\left<S \right>$ is the notation for the group generated by the set $S$. Given a group $G$, a base and a strong generating set for it can be constructed in nearly linear time in the size of the input using the Schreier-Sims algorithm. The orbits $b_i \cdot G^{[i-1]}$ are called the \emph{fundamental orbits}, and the (right) \emph{transversals} $R_i$ are the set of (right) coset representatives for $G^{[i]} \mod G^{[i+1]}$, with the requirement that the representative for the coset $G^{[i+1]} \cdot 1=G^{[i+1]}$ is the identity for all $i$. They can be computed by keeping track of the permutations of $G^{[i-1]}$ that take $b_i$ to each point in the orbit. It follows from the definition of a base that every permutation in $G$ is uniquely defined by its action on the elements of a base, and we can reconstruct the permutation.

Corresponding to the base $B$, we construct the fundamental orbits and the right transversals using the Schreier-Sims algorithm. Initially, we set $\sigma $ to be the identity. We then find the permutation $r_1$ in $R_1$ such that $b_1 \cdot r_1=x_1$. If $x_1$ does not lie in the fundamental orbit of $b_1$, that means no such $\sigma$ exists and we can exit the procedure. We then replace $(x_1,\dots,x_m)$ with $(x_1 \cdot r_1^{-1},\dots,x_1 \cdot r_m^{-1})$. In the $i$-th iteration, we check if $x_i \cdot r_1^{-1}r_2^{-1}\dots r_{i-1}^{-1}$ lies in the orbit of $G_{b_i}$. If it does not, we exit the procedure, else we replace $\sigma $ with $\sigma r_i$; where $r_i$ is the permutation in the transversal $R_i$ that takes $b_i$ to $x_i \cdot r_1^{-1}r_2^{-1}\dots r_{i-1}^{-1}$.
For more details, we refer to Seress~\cite[Chapter 4]{seres}.
\begin{defn}
An uncovering-by-bases (UBB) to fix $r$ errors in a permutation group acting on $\Omega$ is a set of bases for the group, such that, for each $r$-subset of $\Omega$ (a $r$-subset of $\Omega$ is a subset of $\Omega$ of size $r$); there is at least one base in the UBB that is disjoint from the $r$ subset. 
\end{defn}
It is known that every permutation group has an uncovering-by-bases~\cite[Proposition 7]{bailey1}. However, they have been constructed only for a few permutation groups. There is no known procedure to construct a small UBB for an arbitrary permutation group. Uncovering-by-bases can be used to decode as follows:

Consider a permutation $g \in G$ in the list form and let the list $w$ was obtained after $r$ or fewer errors were introduced to $g$. Note that $w$ can have repeated entries and need not be a permutation. Let $R$ be a subset of $\Omega$ of size less than or equal to $r$. Errors were introduced in positions from $R$.  The set $R$ is not known to the decoder a priory. Let $\mathcal{U}$ be a UBB for $G$. It is a consequence of the definition of a UBB that there exists a base $B$ in $\mathcal{U}$ which is completely disjoint from $R$. Hence, in the positions indexed by $B$, the lists $w$ and $g$ are identical. As $B$ is a base, we can reconstruct the entire permutation $B$ from just its action on the points in $B$. While the decoder does not know which base in $\mathcal{U}$ is the appropriate base $B$, it can just repeat this procedure for every base in $\mathcal{U}$. 

This algorithm is similar to the method of \emph{permutation decoding} for linear codes proposed by MacWilliams and Sloane~\cite{macwilliams1964permutation}. This involves using a subset of the automorphism group of the code which moves any errors out of the information positions. 
\subsection{Some constructions of uncoverings-by-bases} \label{examples}
There are only a few examples of UBB that were constructed by Bailey. One of them is a symmetric group acting on $2$-subsets and the other is a wreath product of a regular permutation group with a symmetric group. We talk about them next.
\subsubsection{Symmetric group acting on $2$-sets}
\label{2sets}
Let $\Omega$ be the set of all $2$-subsets of $\{1,2,\dots,m\}$. Consider the symmetric group $S_m$ acting on $\Omega$ with the setwise action: $g \cdot \{i,j\}=\{g \cdot i, g \cdot j\}$. Hence, $G$ is a subgroup of $S_n$, where $n=m(m-1)/2$ in this action. The minimal degree of $G$ is $2(m-2)$ and its correction capacity $r$ is $m-3$~\cite[Proposition 21]{bailey1}.

The set of all $2$-subsets of $\{1,2,\dots,m\}$ can also be viewed as the edge set of the complete graph $K_m$. Thus we can use graph theoretical results to construct bases for $G$, which we will use to construct a UBB. 

A spanning subgraph of $K_m$ which has at most one isolated vertex and no isolated edges forms a base for $G$. A Hamiltonian circuit of $K_m$ is a circuit in $K_m$ containing each vertex exactly once. From a Hamiltonian cycle of $K_m$, we can obtain such a spanning graph (called a V-graph) by deleting every third edge~\cite[Fig.~2]{bailey1}. A Hamiltonian decomposition of a graph is a partition of the edge set of the graph into disjoint Hamiltonian circuits and at most one perfect matching. In the 1890's, Walecki~\cite{walecki} showed that $K_m$ can be decomposed into $(m-1)/2$ Hamiltonian cycles when $m$ is odd. If $m$ is even, it can be decomposed into $(m-2)/2$ Hamiltonian cycles and one perfect matching. Using Walecki's decomposition, we can construct a UBB for $G$. The UBB is the set of all V-graphs which can be obtained from the Hamiltonian decomposition of $K_m$ and it can be proved that any $r$-subset of edges of $K_m$ avoids at least one of these V-graphs. For more on this topic the reader is referred to Bailey~\cite{bailey1}.
This UBB is of size $O(m)$, and each base is of size $O(m)$. Hence, we can use this UBB in the decoding algorithm which would have time complexity $O(m^4)=O(n^2)$.

\subsubsection{Wreath product of regular groups}
Let $H$ be a group acting on a set $\Delta$. It is called a regular group if, for every $x,y \in \Delta$, there exists exactly one $h \in H$ such that $x \cdot h=y$. The permutation code $G$ is the wreath product of $H$, a regular permutation group acting on $m$ points with the symmetric group $S_n$ acting imprimitively on $mn$ points. 

The group $G$ consists of all tuples of the form $(h_1, h_2, \dots, h_n, \sigma)$, where $h_i \in H$ and $\sigma \in S_n$. We use the symbol $1$ to denote the identity element of $H$ and $e$ for the identity element of $S_n$. In this section, we redefine $\Omega=[m] \times [n]$ consisting of $n$ copies of the set $[m]$ as $n$ columns each with $m$ rows. 
The action of $g=(h_1, h_2, \dots, h_n, \sigma)$ on $(i,j)$ is $(i \cdot h_{i \cdot \sigma},j \cdot \sigma)$. The group $S_n$ acts on the columns and the group $H^n$ acts on the rows. As a regular group acting on $m$ points has size $m$, $|G|=m^n n!$. 
\begin{theorem}
The minimal degree of $G$ is $m$ and its error correction capacity is $\floor*{\frac{m-1}{2}}$. 
\end{theorem}
\begin{proof}
Consider the group element $(h,1,\dots,1,e)$ for some $h \in H$. This permutation moves all the points in the first column and fixes all the other points. Hence the minimal degree of $H$ is less than or equal to $m$. We now show that any non identity permutation in $G$ moves at least $m$ points. Consider a permutation $(h_1, h_2, \dots, h_n, \sigma)$. If $\sigma$ is not the identity, it must move all the points in one of the columns to another column and hence moves at least $m$ points. If $\sigma$ is the identity permutation, at least one of the $h_i$ must be a non identity element of $H$ which will move all $m$ points in that column. 
\end{proof}
\begin{theorem}
Any subset of $\Omega$ forms an irredundant base for $G$ if and only if it has exactly one point from each column of $\Omega$. 
\end{theorem}
\begin{proof}
Consider any set $S$ which does not contain any points from the first column. The permutation $(h,1,\dots,1,e)$ fixes every point in $S$. Hence, $S$ cannot be a base for $G$.  
To prove the converse, consider the following set $\{(x_1,1),(x_2,2),\dots,(x_n,n)\}$ with each $x_i$ belonging to the $i^\mathrm{th}$ column. Suppose here exists $g=(h_1, h_2, \dots, h_n, \sigma)$ that fixes each of these points. As $(x_i,i) \cdot g$ must lie in the $i^\mathrm{th}$ column for each $i$, this implies that $\sigma$ must be the identity permutation. Hence, the action of $g$ on $x_i$ is $x_i \cdot h_i$. As the action of $H$ on $[m]$, is regular, this implies that each $h_i$ is the identity. 

We have proved that $B$ is a base for $G$ if and only if $B$ contains at least one point from each column of $\Omega$. If $B$ has more than one point from some of the columns of $\Omega$, we can obtain a subset of $B$ which contains exactly one point from each column of $\Omega$ and is a base for $G$. This means that $B$ must be redundant. Hence, we have proved that $B$ is a irredundant base for $G$ if and only if it has exactly one point from each column of $\Omega$. 
\end{proof}
Using this theorem, we can now construct a UBB for $G$. Each row for of $\Omega$ forms a base and a collection of $r+1$ rows form a UBB for $G$, where $r=\floor*{\frac{m-1}{2}}$. This is because any subset of size $r$ will intersect with at most $r$ of those bases described in the UBB and hence, disjoint from at least one of those bases. 

Having obtained a UBB, we can use it in the decoding algorithm which runs in $O(m^2n^2)$. 
\subsection{An alternative decoding algorithm for wreath product}
\label{altdecoding}
In this section, we present a simple decoding algorithm introduced by Bailey and Prellberg~\cite{bailey3}, which uses majority logic principles for the wreath product groups and compare its performance to the original UBB algorithm. For simplicity, we assume the group $H$ to be the cyclic group $C_m$, although this algorithm can be extended to any regular group. 

Consider a permutation $g$ in $G$, less than $r$ errors were introduced to form $w$. The permutation $g$ acts on $\Omega$, which as noted before, can be viewed of as $n$ copies of the set $[m]$ as $n$ columns each with $m$ rows. First, we rewrite each member of $\Omega$ and the letters in the received word in the form $(i,j)$, with $i \in [m]$ and $j \in [n]$.   
The columns of $\Omega$ form a complete block structure for $G$. As less than $r=\floor*{\frac{m-1}{2}}$ errors were introduced, the majority of positions in each block will contain the correct symbol. Let $g$ be of the form $(h_1, \dots, h_n, \sigma)$ where each $h_i$ is a cyclic shift. In the $j^\mathrm{th}$ column, we consider the $i^\mathrm{th}$ position which contains the tuple $(p,q)$, and calculate the number $p-i$, which corresponds to the cyclic shift in that block. We set the value of $h_j$ to be the most frequently occurring cyclic shift and we set $j \cdot \sigma$ to be the most frequently occurring value of $q$ in that block. When we repeat this procedure for each block, we can obtain the cyclic shifts corresponding to each $h_j$ and obtain $\sigma$ by computing its action on each $j \in [n]$. We now convert $g$ from the form $(h_1, \dots, h_n, \sigma)$ to that of a permutation acting on $\Omega$. 

We now estimate the complexity of this algorithm. The first part of the algorithm involves splitting the word into blocks and rewriting each symbol. This could be done in constant time for each position. Computing the value of the block label and the cyclic shift for each position involves some integer arithmetic which would take a constant amount of time for each position. Hence this would take $O(mn)$ time. 

The second part of the algorithm involves finding the most frequently occurring value of the block value and cyclic shift in each block. Let's consider the part where we find the most frequently occurring value of the cyclic shift. We maintain an auxiliary list of size $m$, with each position corresponding to the frequency of that cyclic shift. We iterate through the block, compute the cyclic shift for each location of the block and update its frequency in the auxiliary list. This procedure takes $O(m)$ time. We then go through the auxiliary list to find the most frequently occurring cyclic shift. We do a similar procedure for finding the action of $\sigma$ on the block using an auxiliary list of size $n$ and would hence take $O(n)$ time. This procedure is repeated for each block.   

The final part of the algorithm, which involves converting the reconstructed word back to a permutation form can be done with $O(mn)$ arithmetic operations. Hence, the algorithm takes $O(mn)$ time if $m >n$ and $O(n^2)$ time if $n >m$. This is faster than the UBB algorithm for the same group, which would take $O(m^2n^2)$ time.
\subsubsection{Decoding more than $r$ errors} \label{analysis}
In fact, the algorithm described above can correct more than $r$ errors, as long as the majority of elements in each block are correct. Hence there are some error patterns with even up to $nr$ errors which can be decoded using this algorithm. Bailey obtained the following expression for the number of error patterns of a certain length which can be corrected by this algorithm (an error pattern is a certain set of positions on which the errors are induced). 

Let $k$ be a positive integer. Let $P_{n,r}(k)$ be the set of all partitions of $k$ into at most $n$ parts where each part is of size at most $r$ . We write such a partition in the form $\pi = \left\{ (i,f_i) \middle| \sum_{i}i f_i =k \right\}$, where $f_i$ is the number of times the number $i$ appears in the partition. For a partition $\pi$, we define the quantity $c_i$ to be $c_i=\sum_{j=1}^{i-1}f_i$, which is the number of parts in $\pi$ of size strictly less than $i$.
\begin{theorem}
For an integer $k \leq nr$, the following number of patterns of $k$ errors can be corrected:
\begin{equation*}
   E_{n,r}(k)= \sum_{\pi \in P_{n,r}(k)} \prod_{i=1}^r \binom{n-c_i}{f_i} \binom{m}{i}^{f_i}
\end{equation*}
\end{theorem}
\begin{proof}
For a proof, we refer to Bailey~\cite[Section 6.7]{bailey_thesis}. It is a simple counting argument using definitions described above. 
\end{proof}
We ran simulations to find the proportion of error patterns which can be corrected for the case $m=5$. In this case, the minimal degree is $5$ and has correction capacity $2$. In practice a much larger number of errors can be corrected, with a high probability of success. 

Using a computer program, we calculate the value of $E_{n,r}(k)$ for different values of $k$ and calculate the value of $\frac{E_{n,r}(k)}{mn}$ to obtain the probability that the alternate decoding algorithm succeeds. We plot this for $n=100$ in Figure~\ref{probdecod}.

We can deduce from the data that the algorithm can decode $95$ percent of error patterns of length up to $19$, $90$ percent of error patterns of length up to $24$, $80$ percent of error patterns of length up to $31$ and $50$ percent of error patterns of length up to $46$. 

\begin{figure}[H]
    \centering
    \includegraphics[scale=0.5]{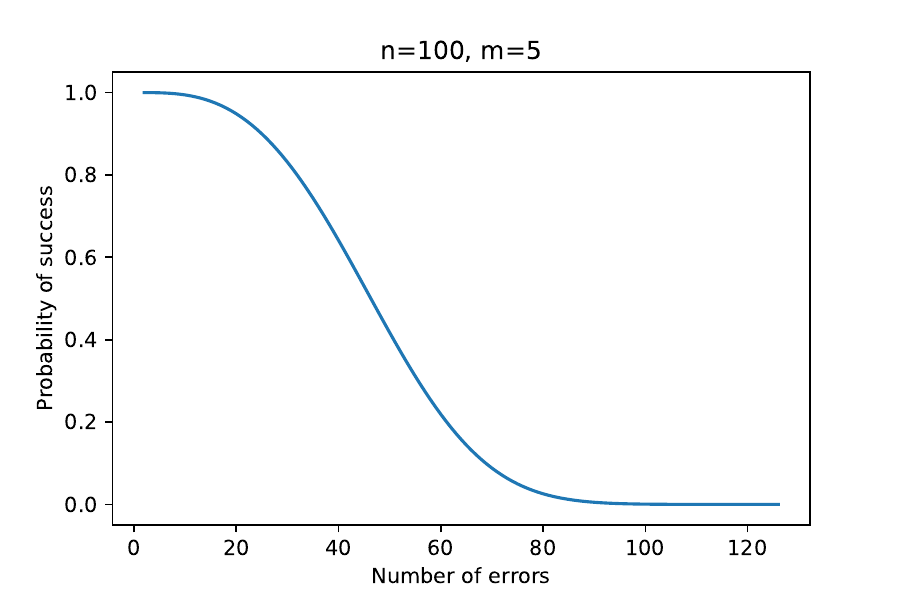}
    \caption{Plot of fraction of errors that can be corrected for $n=100$, $m=5$}
    \label{probdecod}
\end{figure}

Interestingly, with the probability of success fixed, the number of errors which can be corrected increases with $n$. This is despite the minimal degree (correction capacity) remaining the same. We plot the number of errors which can be plotted with $95$ percent probability of success with respect to $n$, both in terms of absolute number of errors which can be corrected and in terms of the \emph{error rate}, which is the ratio of the error induced to the degree of the permutation group. Unlike the correction capacity which increases with $n$, the error rate seems to decrease slightly as $n$ increases. 
\begin{figure}[H]
    \centering
    \includegraphics[scale=0.5]{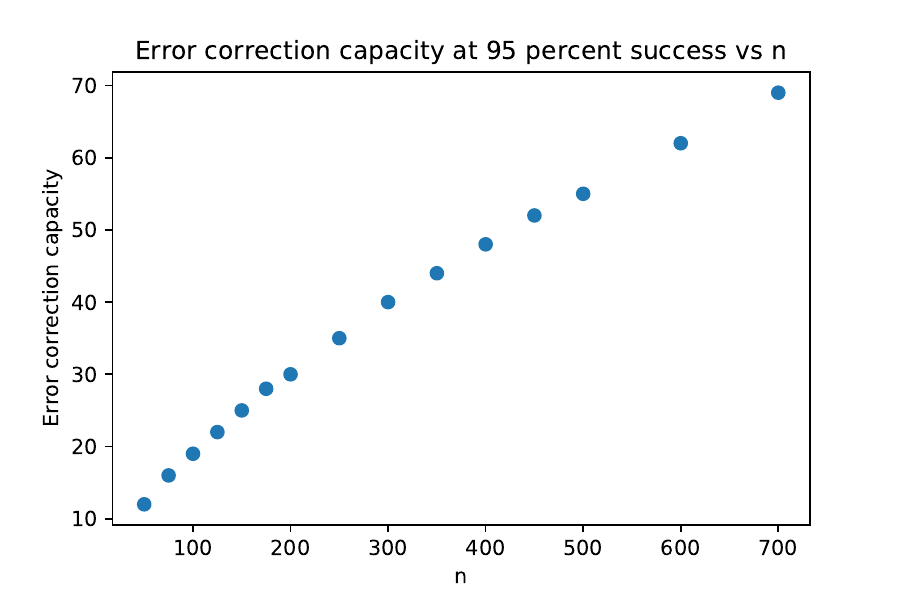}
    \includegraphics[scale=0.5]{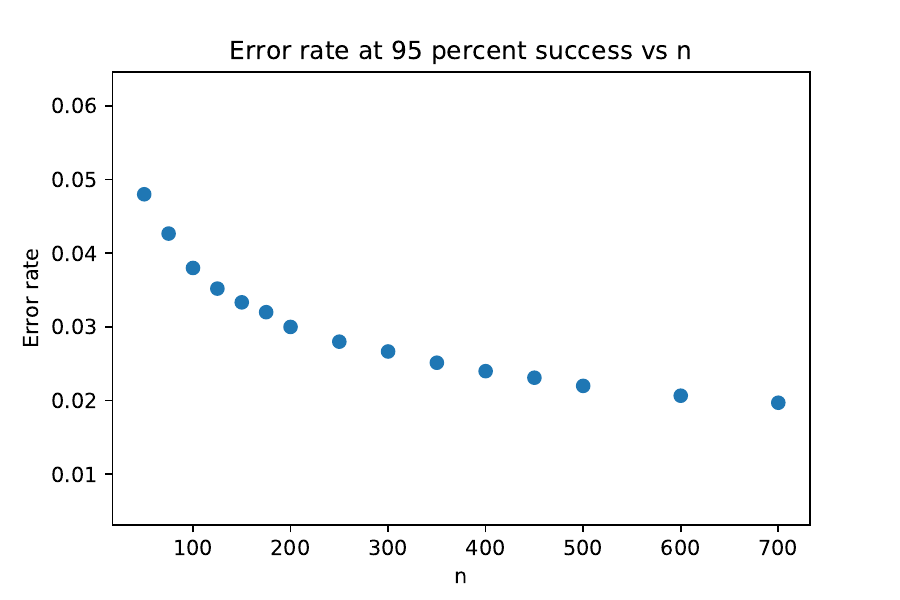}
    \caption{Plot of error correction capacity with $95\%$ probability of success}
    \label{probdecod1}
\end{figure}
\section{A cryptosystem on permutation error-correcting codes}
The McEliece cryptosystem is a very popular cryptosystem using linear codes, first proposed by Robert McEliece~\cite{mceliece1978public} in 1978. Due to its large key sizes, it never gained a lot of popularity at that time. Recently with the rise of quantum computing, it has gained a lot of attention along with its counterpart proposed by Neiderreiter~\cite{niederreiter1986knapsack} as a \emph{post-quantum} alternative to currently used public key cryptosystems. Fundamentally, the McEliece cryptosystem uses two linear codes, one kept private and another made public with the two codes being equivalent in the following way:
\begin{defn}
    Consider two $[n,k]$ linear codes $C_1$ and $C_2$ generated by the generator matrices $G_1$ and $G_2$. These codes are considered equivalent if, there exists a non singular $k \times k$ matrix $S$ and an $n \times n$ permutation matrix $P$ such that $G_2=S G_1 P$
\end{defn}
Consider a linear code $C$ generated by a matrix $G$. Pre-multiplying $G$ by a non-singular matrix will not change the code. Post-multiplying $G$ by a permutation matrix creates a new code, one whose codewords are obtained by permuting the positions of the codewords in $C$ by the permutation corresponding to the permutation matrix. In other words, there is a \emph{distance preserving map}, or isometry between two equivalent linear codes. The principle behind the McEliece cryptosystem is that we have a map between an easy instance and a hard instance of the closest vector problem in linear codes using this distance preserving map. The private key matrix, is chosen from a family of linear codes with a fast decoding algorithm. However, the public key matrix generates a linear code do not have a fast decoding algorithm. The map between these easy and hard instance is known only to the owner of the private key. This provides the one-way property which is essential to all public-key cryptosystems.

In the previous section, we studied permutation groups as error-correcting codes. In this section, we investigate the following questions:
\begin{itemize}
    \item Can we develop a public key cryptosystem similar to the McEliece cryptosystem using permutation codes instead of linear codes?
    \item Is there any potential advantage of using permutation codes instead of linear codes from the standpoint of cryptography?
\end{itemize}
The first step in this direction would be to investigate isometries of the symmetric group with the Hamming metric. Assume that $\phi$ is an isometry of the Hamming space. We can have the private key as a permutation code $H$ (which is a subgroup of $S_n$) with an efficient decoding algorithm and the public key as the permutation code $\widehat{H}=\phi(H)$ for which the decoding algorithm will not work. The isometries in the symmetric group were studied and classified by Farahat~\cite{isometry} who proved that they are of one of the following types: 
\begin{itemize}
    \item $x \mapsto g x h$ for $g$ and $h$ being fixed members of $S_n$
    \item $x \mapsto g x^{-1} h$ for $g$ and $h$ being fixed members of $S_n$
\end{itemize}
However the isometries that interest us should be homomorphisms. This is due to practical reasons. We want both the codes $H$ and $\widehat{H}$ to be subgroups of $S_n$. This is because a permutation group can be represented efficiently using a generator set. A permutation array without a group structure cannot be represented efficiently this way. Hence, we are interested in isometries of $S_n$ that are homomorphisms. We now combine these ideas to create a public-key cryptosystem. Our ingredients are a permutation code with a fast decoding algorithm and a conjugation map. 
    \subsection{A McEliece cryptosystem using permutation codes}
	In this section, we present a public-key cryptosystem using permutation codes that is similar to the McEliece cryptosystem. The private key is kept secret, the public key is available to the public.
	\begin{description}
	\item[Private key] Let $G$ be a permutation group acting on $n$ letters with a fast decoding algorithm. We choose a subgroup $H\leq G\leq S_n$ and a permutation $g$ from $S_n$. The private key is a set of  generators $\{ h_1,\ldots,h_s \}$ of $H$ and $g$. We will use the decoding algorithm in $G$ as a decoding algorithm for $H$. 
	\item[Public key] The public key is a set of elements $\{ \hat{h}_i \}=\{ g^{-1}h_ig \}$ which generate the conjugate group $\widehat{H}=g^{-1}Hg$. Note that a security assumption is that $\widehat{H}$ does not have a fast decoding algorithm. Note that $\widehat{H}$ might not be a subgroup of $G$. 
	\item[Encryption] The plaintext is $\hat{h}$ of $\widehat{H}$. We introduce $r$ errors in $\hat{h}$ to create the ciphertext $c$. 
	\item[Decryption] Once we receive $c$, we compute $gcg^{-1}$. We then use the fast decoding algorithm to obtain $h$, which is its nearest neighbour in $H$. We then compute $\hat{h}=g^{-1}hg$. 
	\end{description}
	The conjugation map is an isometry of the Hamming space. That is why this cryptosystem works. The distance between $h$ and $H$ is equal to the distance between $\hat{h}$ and $\widehat{H}$. In the McEliece cryptosystem, the Hamming isometry is post-multiplication by a permutation matrix. The secret scrambler scrambles the basis of a linear code to create a new basis. It does not change the linear code. Here, the analogue to that would be the fact that we choose a random set of generators for the conjugate group as the public key.
	
\paragraph{Remark:}There is a possibility of using isometries other than the conjugation map. However, the errors have to be introduced more carefully. For example, suppose we have a group $G \leq S_n$ and a map from $G$ to $G$ which is an isometry on $G$. Our permutation code is a subgroup $H \leq G \leq S_n$. Hence, if we introduce errors in $H$ so that the ciphertext $c$ lies in the group $G$, the cryptosystem using this isometry will work. We have not studied this problem carefully and leave it for future work. 

	\subsection{The security of our cryptosystem}
	\subsubsection{Security assumptions}
	There are two security assumptions:
	\begin{description}
		\item[AS1] Solving the general decoding problem in $\widehat{H}$ is hard. This means that $\widehat{H}$ has no good decoding algorithm. 
		\item[AS2] Inability to construct a subgroup $H^\prime$ of $S_n$ with generators $\{h^\prime_1,\ldots,h^\prime_s\}$ and $g \in S_n$, such that $g^{-1}H^\prime g=\widehat{H}$ and there is an efficient decoding algorithm in $H^\prime$ that can correct up to $r$ errors. 
	\end{description}
	We now discuss possible attacks on our cryptosystem and its security in the face of these attacks. Like the McEliece cryptosystem, attacks on this cryptosystem can broadly be classified into two types -- unstructured and structural attacks. Unstructured attacks attempt to solve the nearest neighbour problem in $\widehat{H}$ directly without attempting to obtain the private key. A structural attack would attempt to obtain the private key from the available public information. 
	
	The group $H$ has an efficient decoding algorithm. For the cryptosystem to be secure, the group $\widehat{H}$ must not have an efficient decoding algorithm. Specifically, if the algorithm used to decode in $H$ can be modified to work for the conjugate group $\widehat{H}$, the cryptosystem will not be secure. This would be another type of attack which would be specific to the cryptosystem using a particular permutation code. Hence, the security of the cryptosystem depends on the structure of the permutation code $H$ used and its decoding algorithm, and not just on parameters like size, correction capacity, etc. This is especially true in the context of structural attacks. Just like the case with the McEliece cryptosystem, we cannot provide any theoretical basis to the claim that the cryptosystem using a certain permutation code is secure. Rather, we construct a cryptosystem and then attempt to come up with attacks to demonstrate its security. Due to the similarity of our cryptosystems with linear code based cryptosystems, which are very well studied and believed to be secure, one good place to start would be to try and extend the well known attacks on the McEliece cryptosystem to our cryptosystem. One of the attacks we consider in this section is the well studied and powerful information set decoding (ISD) attacks on the McEliece cryptosystem. For the case of permutation codes, we have designed an algorithm which is very similar in spirit to the ISD attack and we study its effectiveness.
	
	We now present some generic attacks on our cryptosystem. These attacks work regardless of the choice of the permutation code. Hence, their effectiveness only depend on public parameters like the correction capacity and size of the code. Later, we will demonstrate the existence of a code with parameters that make the cryptosystem secure against these attacks (Section~\ref{wreath}).  
	\subsubsection{Brute Force Attacks}
	One obvious method to solve the nearest neighbour problem in $\widehat{H}$ is to enumerate all elements of $\widehat{H}$ and calculate their hamming distance from the ciphertext. This means that $\widehat{H}$ must be large. Its exact size requirement would depend on the level of security we are looking for. 
	
	Another way is to enumerate all possible permutations within a distance $r$ from the codeword $c$ and check if they belong to $\widehat{H}$. The complexity of this attack would be $\binom{n}{r}(n-1)^r$, where $n$ and $r$ are the degree and correction capacity of the group respectively. 
	\subsubsection{Information Set Decoding}
	One of the more successful attacks on the McEliece cryptosystem uses the technique of information set decoding. The original idea was proposed by Prange~\cite{prange} in 1962 and over the years many improvements on this basic attack have been proposed~\cite{info1,bernstein2008attacking}. For an $[n,k]$ linear code, it involves picking $k$ coordinates at random and trying to reconstruct the codeword from those positions. We now recreate a version of ISD for permutation codes:
	
	\begin{algorithm}[H]
	\KwIn{generators of $\widehat{H}$, a ciphertext $c$}
	\SetKwFunction{FSum}{ISD}
    \SetKwProg{Fn}{Procedure}{:}{}
    \Fn{\FSum{}}
    {
	Choose a base $B$ and SGS for $\widehat{H}$ arbitrarily\\
	\If{There are no repeated symbols in positions indexed by $B$}{
	Use element reconstruction algorithm to try and find an element $h \in \widehat{H}$ agreeing with $c$ on these positions.\\
	\If{$h$ exists \textbf{\textrm{and}} $d_H(h,c) \leq r$}
	{\Return $h$
	 }
	}
	\Return \textbf{False}
	}
	Procedure repeated till it succeeds
	\caption{Information Set Decoding}
	\end{algorithm}
	This algorithm is similar to the uncovering-by-bases decoding algorithm proposed by Bailey, except that, we do not know the UBB and hence, the running time is not bounded by the size of the UBB. The correctness of this algorithm follows from the proof in Bailey~\cite[Proposition 7]{bailey1} in which he shows that given any $r$-subset of $\{1,\ldots,n\}$, there exists a base for $H$ disjoint from it. Each iteration is a success if none of the positions of $c$ indexed by the chosen base have an error. This occurs with a finite non-zero probability because such a base exists. The performance of the ISD algorithm is identical for the groups $H$ and $\widehat{H}$. This is because $B$ is a base for $H$ if and only if $g^{-1}B$ is a base for $\widehat{H}$, and the algorithm picks a base at random where $g$ is the secret conjugator.
	
	\paragraph{Complexity} The algorithm uses a procedure to choose an arbitrary base $B$ for the group. At the same time, a set $R$ of size $r$ is chosen uniformly at random and errors are induced in those positions. Before making concrete complexity estimates, we would like to make some observations. First, it is clear that if there are more error positions, the probability that the information set does not have any errors is lower and hence the complexity of the attack increases with $r$. Also, if the size of the base chosen is small, the probability that the positions indexed by this base does not have any error positions is high. Hence the complexity of the attack increases with $k$, the expectation of the length of the base.
	
    Let $E_B$ denote the event that the algorithm chooses subset $B$ as a base for the group. Conditioned on the event $E_B$, an iteration of the algorithm succeeds if the sets $R$ and $B$ are disjoint. Because the set $R$ is chosen uniformly at random from $\{1,\dots,n\}$,
	\begin{equation*}
	    Pr[\text{Success}|E_B] \geq \frac{\binom{n-|B|}{r}}{\binom{n}{r}}.
	\end{equation*}
	As this probability depends only on the size of $B$ and not on the set $B$, we can obtain an expression for the probability of an iteration of the algorithm succeeding:
	\begin{equation*}
	     Pr[\text{Success}]\geq\sum_k  Pr[\text{Success}||B|=k]Pr[|B|=k]=\sum_k \frac{\binom{n-k}{r}}{\binom{n}{r}}Pr[|B|=k].
	\end{equation*}
	Let $k_{\text{max}}$ be the size of the irredundant base of largest size of $H$. As the $Pr[\text{Success}|E_B]$ decreases as the size of the base increases, we can obtain the following lower bound on the probability of the algorithm succeeding:
	\begin{equation*}
	    Pr[\text{Success}]\geq \frac{\binom{n-k_{\text{max}}}{r}}{\binom{n}{r}}
	\end{equation*}
	This is a Monte Carlo algorithm, which decodes correctly with probability $Pr[\text{Success}]$, and otherwise returns false. If we repeat the algorithm till it succeeds, we obtain a Las Vegas algorithm that always decodes correctly but has a running time which is a random variable with expectation $1/Pr[\text{Success}]$. It is also important to note that as each iteration is independent, this attack can be effectively parallelized. Using the following well known formula for binomials, 
	\begin{equation*}
	    \binom{n}{r}=\Theta\left(2^{H_2(\frac{r}{n})n}\right) 
	\end{equation*}
	we can show that this algorithm has expected complexity: 
	\begin{equation} \label{seceqn}
	    \Theta\left( 2^{\alpha(k_{\text{max}},r)n } \right)
	\end{equation}
	where $H_2(p)=-p \log_2 p -(1-p) \log_2(1-p)$ is the binary entropy function and $\alpha(k,r)=H_2\left(\frac{r}{n}\right)-\left(1-\frac{k}{n}\right)H_2\left(\frac{r}{n-k}\right)$. The running time increases with both error rate and information rate. Gill and Lod\`{a}~\cite{gill} computes $k_{\text{max}}$ for the symmetric group.
	
	The security level of a cryptosystem is measured using the amount of computational resources needed to break it. It is often expressed in terms of `bits', where a cryptosystem that is $b$-bit secure requires $2^b$ operations to be broken. For asymmetric cryptosystems, this security level is obtained using the running time of the best known attacks on them. Using the ISD, attack,  we can now obtain the following necessary condition for our cryptosystem to have $b$-bit security:
	\begin{equation}
		\alpha(k_{\text{max}},r)n \geq b
	\end{equation}
	Hence, if we want a cryptosystem with a $b$-bit security requirement, we need to choose a permutation code with the appropriate parameters.  
	
	So far, we have described a framework to develop cryptosystems using permutation codes and outlined generic attacks on those cryptosystems, whose effectiveness depend on the parameters of the chosen permutation code $H$. Next, we attempt to use the permutation codes described in Section~\ref{examples} in our cryptosystems. 
	
	\subsection{Our cryptosystem using wreath product groups}
	\label{wreath}
	The security and the performance of our cryptosystem depends on the choice of the group $H$ and its decoding algorithm. In this section, we explore using the wreath product groups with the alternate decoding algorithm discussed in Section ~\ref{altdecoding} for this purpose. We show that, using the appropriate parameters, it can be made resistant to information set decoding attacks. The decoding algorithm can be modified to work in the conjugate group $\widehat{H}$ as well, which would make the cryptosystem using it insecure. Although this cryptosystem is insecure, it does provide a concrete example of a permutation code with an efficient decoding algorithm which cannot be efficiently decoded using the ISD algorithm. 
	
	Let $H$ be the wreath product group $C_m \wr S_n$ (more generally, we can take $H$ to be a subgroup of $C_m \wr S_n$ with suboptimal decoding). We consider the case of $m=5$ for different values of $n$, for different probabilities of correct decoding. Because the actual minimal degree of this group is $5$, there will be cases in which the receiver of the encrypted word cannot decode correctly and will receive a wrong word. In these cases, the receiver will be able to tell that the decoding is wrong (possibly by using an appended hash function) and ask the sender the resend the message. Consider the case for which the probability of correct decoding is $0.9$ and $0.95$. We find the largest such value of $r$ for which $r$ errors are introduced and can be corrected for at least $0.9$ or $0.95$ fraction of the cases for different values of $n$, and estimate the amount of computational resources required for an ISD based attack to break the cryptosystem by plugging in this value of $r$ into Equation~\ref{seceqn}.
	\begin{figure}
	    \centering
	    \includegraphics[scale=0.5]{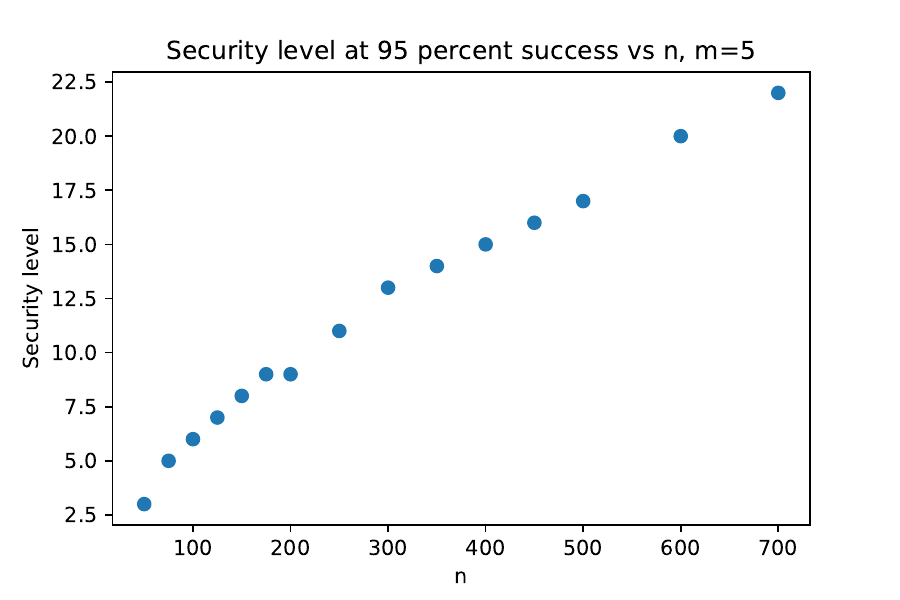}
	    \includegraphics[scale=0.5]{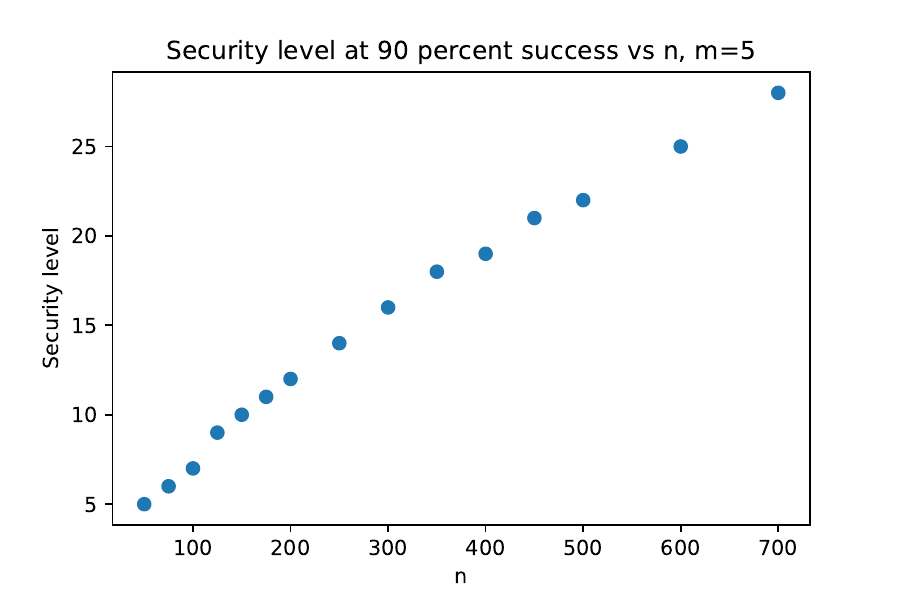}
	    \caption{Security levels for $m=5$}
	    \label{fig:my_label1}
	\end{figure}
	As we can see, the computational resources needed to break this cryptosystem is non trivial and increases with $n$, although it is still not close to commercial levels of security. 
	We also repeat this exercise for the case of $m=7$.
	\begin{figure}
	    \centering
	    \includegraphics[scale=0.5]{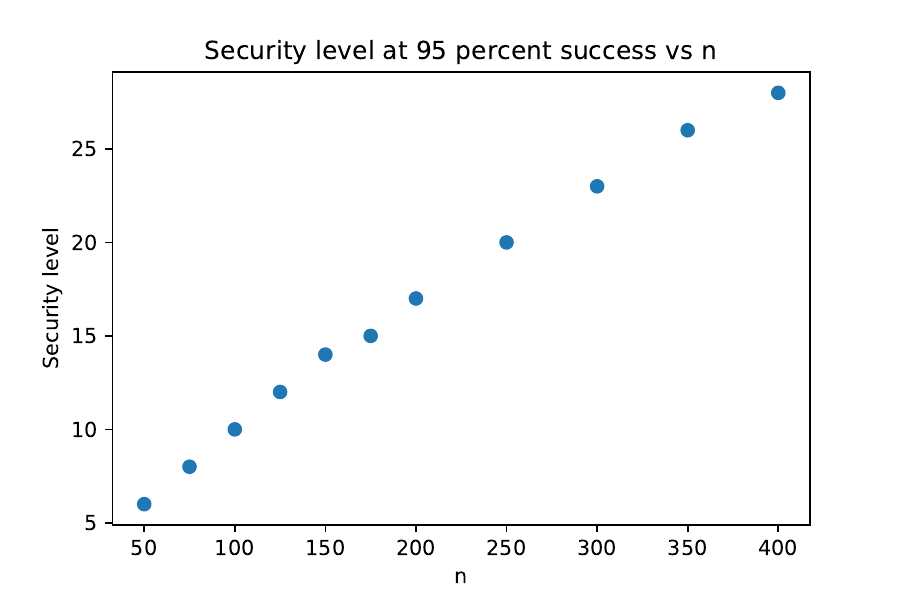}
	    \includegraphics[scale=0.5]{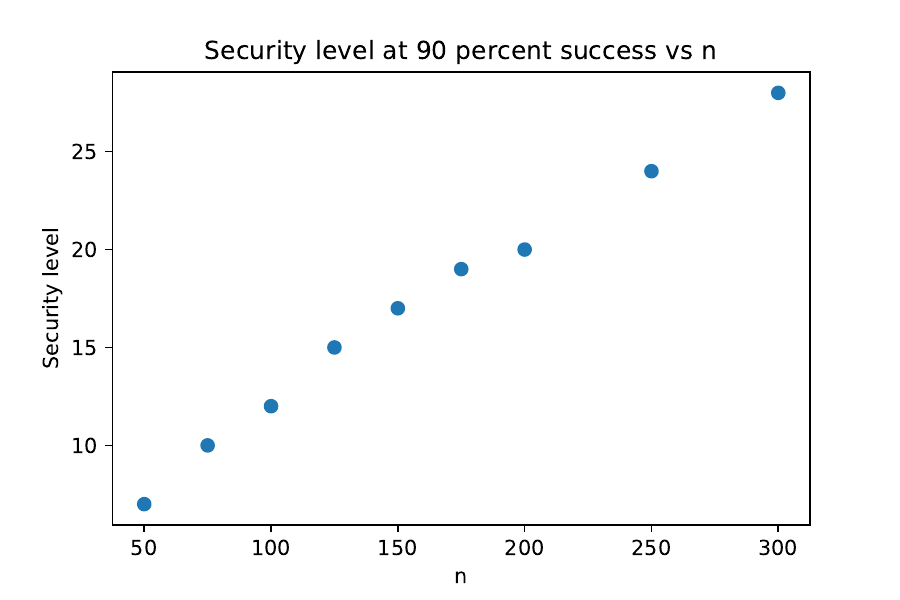}
	    \caption{Security levels for $m=7$}
	    \label{fig:my_label2}
	\end{figure}
	For $m=7$, we can attain higher security levels for smaller values of $n$ compared with $m=5$. These results mean that, by increasing $m$ and $n$ to an appropriate amount, we can reach commercial levels of security such as $128$ or $256$ bits. Our analysis of larger values of $m$ and $n$ were limited by our computational resources, as our computer program for calculating the proportion of error patterns that can be corrected involves iterating over integer partitions, which is computationally intensive. Judging by the trend in the graphs however, we can conclude that the values of $m$ and $n$ can be appropriately increased to attain larger levels of security.
	\subsubsection{Attack using block systems}
	For a permutation group $G$ acting on $\Omega$, a complete block system $\mathcal{B}$ is a partition of $\Omega$ into disjoint sets $B_1, \dots, B_k$ called blocks, such that if two points $x$ and $y$ are in the same block, then for all permutations $g \in G$, $x \cdot g$ and $y \cdot g$ are also in the same block. It can be shown that every block in $\mathcal{B}$ have the same cardinality if $G$ is transitive. 
	
	The first property of block systems that we shall prove is that conjugation preserves the block structure of a permutation group. 
	\begin{theorem}
	Consider a permutation group $G$ with a block system $\mathcal{B}=\{B_1, \dots, B_k\}$. The conjugate permutation group $G'=\sigma ^{-1}G\sigma$ has the block system $\mathcal{B}'=\{B_1 \cdot \sigma, \dots, B_k \cdot \sigma\}$.
	\end{theorem}
	\begin{proof}
	Every permutation in $G'$ is of the form $\sigma^{-1} g \sigma$ where $g$ is a member of $G$. Consider any two points $x$ and $y$ in the set $B_i \cdot \sigma$. Then
$	    x \cdot \sigma^{-1} g \sigma=x' \cdot g \sigma \text{ and }y \cdot \sigma^{-1} g \sigma=y' \cdot g \sigma.$
	As $x'$ and $y'$ are in the same block $B_i$, $x' \cdot g$ and $y' \cdot g$ are in the same block of $\mathcal{B}$, which means that $x \cdot \sigma^{-1} g \sigma$ and $y \cdot \sigma^{-1} g \sigma$ are in the same clock of $\mathcal{B}'$ 
	\end{proof}
	For the wreath product group $C_m \wr S_n$ acting on a set of $mn$ points it is easy to see that the columns form a maximal block system, with the group $C_m$ acting on the blocks and the group $S_n$ permuting them. Hence, the public key group $\sigma^{-1}G\sigma$ has a corresponding block structure. First, we observe that the alternative decoding algorithm to be modified for a more generic type of group, with a block system such that the group acts regularly on each block.   
	
	Consider any permutation group acting on a set $\Omega$ with a block system consisting of $n$ blocks of size $m$ each. It is well known that this permutation group can be described as a subgroup of the wreath product group $S_m \wr S_n$. Every permutation is described by how it acts on each block and by how it permutes the blocks among themselves. In this case the permutation group also acts regularly on each block.
	
	Let the conjugate of $G$ be $G^\prime$. We compute a maximal block system for $G^\prime$ (or, if we use a subgroup of $G$, we can look for a block system with blocks with the required size). We then compute a set of generators for a subgroup of $G^\prime$ which stabilizes this block system. This subgroup will act regularly on each block. We call these regular subgroups $H_1,H_2,\ldots,H_n$. An element of a regular subgroup is uniquely determined just by its action on one point. This element can also be computed using the Schreier-Sims algorithm.
	
	Hence, for each position in $\Omega$, we compute the block labels and relabel the received word accordingly. For each position, we then compute the element of the corresponding regular group and the value of the block permutation. We then compute the most frequently occurring value of these elements in each block to obtain a group element of the form $(h_1, h_2, \ldots, h_n, \sigma)$ where $h_i \in H_i$ and $\sigma \in S_n$ and convert it back to permutation form. Hence, the same decoding algorithm can be deployed on the public key group, breaking the cryptosystem.
	\subsection{Our cryptosystem using the symmetric group acting on $2$-subsets}
	We now examine the use of a group whose decoding algorithm will not work on its conjugates. Because of its parameters  however, it can be trivially broken using ISD attacks. 
	
	Let $G$ be the symmetric group $S_m$ acting on the $\Omega$, defined as, the set of  all two subsets of $\{1,\ldots,m\}$. Let $n=\frac{m(m-1)}{2}$ and we choose $H$ to be a subgroup of $G$. However, $H$ is also a subgroup of $S_n$. For that there is a  relabeling of the two subsets of $\{1,2,\ldots,m\}$ by elements of $\{1,2,\ldots,n\}$. This labeling is not publicly available. Thus publicly we see $H\leq G\leq S_n$.
	
	We use the decoding algorithm using the UBB described in Section~\ref{2sets}. As $H$ is a subgroup of $G$, the same decoding algorithm can be used, although the decoding might be suboptimal because the minimal degree of $H$ might be greater than $G$. Unlike the case of the alternative decoding algorithm for the wreath product groups we cannot easily construct a UBB for the conjugate group $\widehat{H}\leq S_n$ using the same techniques used to construct one for $G$. To recall, the UBB is constructed using the Hamiltonian cycles of a complete graph. Now the conjugation map is just relabeling. When this relabelling happens, the V-graphs changes arbitrarily making them not bases. There seems to be no way of constructing a UBB without finding the conjugator. Furthermore, note that, since $H$ is a subgroup of $G$, there is insufficient information on the action of the $G$ on the complete graph from which the UBB is constructed.
	
	The next thing we attempt to do is attack the second security assumption. That is, we attempt to find a conjugator $g$ and a subgroup $H^\prime$ of $S_n$ with an efficient decoding algorithm. We assume that these subgroups are all subgroups of $G$, as we can employ the UBB algorithm for them. We do manage to come up with an attack which is better than a brute force search for $g$.
	\subsubsection{Conjugator Search Attack} 
	To recall, we are given a sequence of generators $\{\hat{h}_i\}$ for the group $\widehat{H}$. We attempt to construct a set of generators for a subgroup $H^\prime$ of $G$ and $g\in S_n$ such that $g\widehat{H} g^{-1}=H^\prime$. We provide a sketch of the attack:
	
	Consider any element of $S_m$. Its cycle type in the permutation image acting on $2$-subsets is entirely determined by its cycle type acting naturally. It is possible for two different cycle types acting naturally to lead to the same cycle type acting on $2$-sets. Two elements which are conjugate in $S_m$ are also conjugate in $S_n$, but two elements can be conjugate in $S_n$ but not in $S_m$.
	
	Consider the elements $h_1$ and $\hat{h}_1$. The crucial thing to note is that it is enough to find the cycle type of the pullback $\hat{h}_1$ in $S_m$. This is because any element of the same cycle type of $h_1$ is conjugate to $h_1$ in $S_m$, and we can conjugate all the other generators also by the same element to obtain a different subgroup $H^\prime$ of $S_m$ which would also have an efficient decoding algorithm.
	
	Let us say we find an element of $S_m$ with the same cycle type as $h_1$ (cycle type in $S_n$). The conjugator is unique up to a coset of the centralizer of $\hat{h}_1$ in $S_n$. Among all the elements in this coset, we need to find a conjugator that takes each $\hat{h}_i$ into $S_m$. The only way to do that (to our knowledge) would be by an exhaustive search. The bottleneck in the algorithm is this step. Finding $g \in S_n$ that conjugates an element $a \in S_n$ to an element $b \in S_n$ is not considered a difficult problem in practice and neither is computing the centralizer of an element as there are very effective backtracking methods to solve these problems~\cite[Section 9.1.2]{seres}. Thus taking $g\in S_n$ does not help with the security of the cryptosystem.
\subsubsection{ISD attacks}
	The size of a base for $H$ is less than $\log_2(m!) \leq m \log_2 m$. Hence, the number of bit operations required to break the cryptosystem using information set decoding is less than: 
	\begin{equation*}
	    \left( H_2\left( \frac{m-3}{m(m-1)/2}\right)-\left(1-\frac{m \log_2 m}{m(m-1)/2} \right) H_2\left( \frac{m-3}{m(m-1)/2-m \log_2 m} \right)\right)\frac{m(m-1)}{2}
	\end{equation*}
	We plotted this as a function of $m$ (Fig.~\ref{security}) and we can see that the security level of even $80$ bits isn't attainable for reasonable values of $m$.
	
	The key reason why the cryptosystem using this code is insecure is the \emph{parameters} of the code. Both the rate of the code and the correction capacity are too low. What we are aiming for is a code where the quantities $k/n$ (on average) and $r/n$ are asymptotically constant, like the binary Goppa codes. This requirement is satisfied by the wreath product groups using the alternative probabilistic decoding algorithm. 
	\begin{figure}[ht] 
	    \centering
		\includegraphics[scale=0.5]{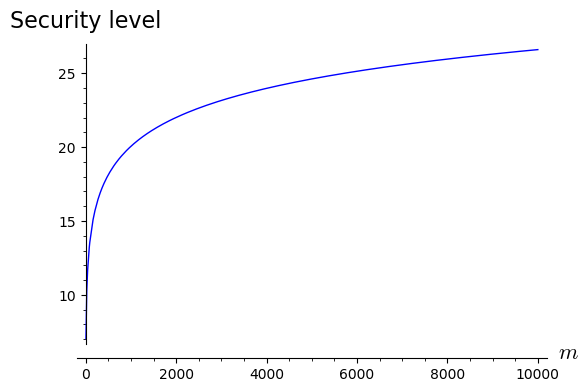}
		\caption{Security of cryptosystem using $2$-sets}
		\label{security}
	\end{figure}
	\subsection{An enhanced McEliece cryptosystem}\label{enhanced}
	Permutation codes are in permutation groups which have less structure than linear codes which are vector spaces. Linear codes can also be seen as permutation codes.
Let $q=p^t$ where $p$ is a prime and $t$ a positive integer. A linear code of block size $n$ over the field $\mathbb{F}_q$ is an additive subgroup of $\mathbb{F}_q^n$ and a permutation code of block size $n$ is a subgroup of $S_n$. We now \emph{enhance the classical McEliece cryptosystem over linear codes} using permutation codes. For this we describe an embedding. 
  
	We start with a monomorphism from the abelian group of $\mathbb{F}_q$ to the symmetric group $S_{tp}$ where $q=p^t$ and $p$ is a prime. The representation for $\mathbb{F}_q$ is the polynomial representation. This means that $\mathbb{F}_q\equiv \mathbb{F}_p[x]/\phi(x)$ where $\phi(x)$ is a irreducible polynomial of degree $t$. This says that every element of $\mathbb{F}_q$ is a polynomial $\alpha_0+\alpha_1 x + \alpha_2 x^2 +\cdots+\alpha_{t-1}x^{t-1}$. Now we fix $t$ disjoint cycles $c_0,c_1,\ldots,c_{t-1}$ of length $p$ in $S_{pt}$ and order them. One can choose $c_0=(1,2,\ldots,p)$ as the first cycle, $c_1=(p+1,\ldots,2p)$ as the second and so on. Note that disjoint cycles commute. Now we define a map, which we call the \textbf{basis map}, from the additive group of 
	$\mathbb{F}_q$ to $S_{pt}$:
	\begin{align}
	0\mapsto (\,),\; \alpha_0\mapsto c_0^{\alpha_0}\;\text{and}\;
	\alpha_ix^i\mapsto c_i^{\alpha_i}\; \text{for}\; i =\{1,2,\ldots,t-1\}.\nonumber
\end{align}	 	
Thus $\alpha_0+\alpha_1 x + \alpha_2 x^2 +\cdots+\alpha_{t-1}x^{t-1}\mapsto 
c_0^{\alpha_0}c_1^{\alpha_1}\cdots c_{t-1}^{\alpha_{t-1}}$. It is easy to check that this map is an embedding -- a monomorphism. Note that the map depends on the choice of the cycles.
Now if we have a vector space $V$ of dimension $n$ over the field $\mathbb{F}_q$, corresponding to a fixed ordered basis $\mathcal{B}$ of $V$, every element of $V$ is a vector of size $n$ over $\mathbb{F}_q$ -- the co-ordinate vector. These co-ordinate vectors depends on the fixed ordered basis of the vector space. Then the embedding of this vector space in the symmetric group $S_{pt}^n$ is done in the obvious way. A vector in $V=\langle v_1,\ldots,v_{n}\rangle$, where the basis is $\mathcal{B}=(v_1,\ldots,v_n)$, is of the form $\nu_1v_1+\ldots+\nu_{n}v_{n}$ where $\nu_i$ are field elements -- the coordinates. Then $v$ corresponds to the ordered tuple $(\nu_1,\ldots,\nu_{n})$ of field elements. This ordered tuple is then mapped to an ordered tuple $(\hat{\nu}_1,\ldots,\hat{\nu}_n)$ in $S_{pt}^n$ a direct product of $n$ copies of $S_{pt}$ in the obvious way, using the basis map already defined. 
This embedding respects addition in the field and then in the vector space. It also respects scalar multiplication as long as the scalar is in the prime subfield.
Thus a linear code $S$ being a subspace of the vector space $V$ can be embedded as an elementary abelian subgroup of the symmetric group -- a permutation group. 
	
Let $S=\langle s_1,s_2,\ldots, s_d\rangle$ be a $d$-dimensional subspace of a $n$-dimensional vector space $V$ which is a code with a good decoding algorithm. This code has a $d \times n$ generator matrix. This basis matrix depends on a fixed ordered basis $\mathcal{B}$ of $V$. One can define a classical McEliece cryptosystem based on this code $S$. As in a McEliece cryptosystem we define the public key as $S^\prime = ASP$ where $A$ and $P$ are the scrambler and permutation matrices respectively. Now $S^\prime$ is also a $d\times n$ matrix over $\mathbb{F}_q$. One can use the basis map defined earlier to transfer this matrix to a matrix $S_1$ over $S_{pt}$. Here $S_1$ is constructed by taking the embedding of each element in $S^\prime$. Then each element of $S_1$ is conjugated by $g$ an element of the symmetric group $S_{pt}$ and denoted by $S_1$ as well. Thus $S_1=g^{-1}S_1g$. Recall that this conjugation is just a relabelling of the permutations.

In this enhanced McEliece cryptosystem the public key is the rows of the matrix $S_1$. The whole structure of the classical McEliece cryptosystem $S$ and $S^\prime$ is secret and so is the basis map and the conjugator $g$. The plaintext is $[a_1,\ldots,a_n]$ where each $a_i$ is in the $\mathbb{F}_p$ the prime subfield of $\mathbb{F}_q$. Then one computes $a=\prod\limits_{i=1}^n (S_1[i])^{a_i}$ where $S_1[i]$ is the $i$th row of $S_1$. Then we introduce errors in this permutation $a$. This is the ciphertext.

On receiving the ciphertext, use the conjugator $g$ to restore original labelling. Then one decomposes it as product of cycles. Then one computes most of the $a_i$ from the exponent. Then that gives rise to a vector of length $n$ over $\mathbb{F}_q$ via the embedding. Assume that this is the cyphertext for the classical McEliece cryptosystem. Decrypt it using the classical McEliece cryptosystem. We will get the plaintext $[a_1,\ldots,a_n]$ back.

There are few things to note. First, since the basis map is a secret it enhances the classical McEliece cryptosystem. Thus one might be able to use linear codes for the enhanced McEliece cryptosystem that is otherwise not secure. When it comes to cycles, computing the exponent has a lot of redundancy. Say, in a cycle after introducing errors there is one point that gets repeated. Then we can simply ignore the repeated points and compute the exponent. It is the action on one point that determines the exponent, when it comes to a cycle. Furthermore, a field element now has a permutation representation. Based on this representation we might be able to introduce much more errors than is possible to fix for a permutation error correcting code. This topic, what is the right way to introduce errors and how many errors we can introduce is left unsettled in this paper. Lastly, there is no need for the matrix $S^\prime$. One can just move from $S$ to $S_1$ directly without using $S^\prime$. In this case use a scrambler matrix over the prime subfield. Then the extra step in decoding can be avoided.

Assume now that there is a classical McEliece cryptosystem over linear codes that is secure. Then for a suitable ciphertext of the original McEliece cryptosystem, one can map it to $S_1$ by the basis map. If the original plaintext is recovered from the ciphertext in the permutation setting then that breaks the original McEliece cryptosystem. Note, we assume that the basis map is known in this case. Thus this McEliece cryptosystem is truely an enhanced version of the classical McEliece cryptosystem and is bit more than mere academic interest and deserves further study.
\subsection{Why use permutation codes?}
	So far, we have described a framework for using permutation codes in public key cryptography and explored this framework using two particular permutation codes. The cryptosystem using wreath product groups can be made resistant to ISD attacks but its decoding algorithm can be adapted for use in any of its conjugates too, which makes it unsuitable. On the other hand, for the symmetric group acting on $2$-sets, the decoding algorithm cannot be modified for use in its conjugates, nor can the conjugator be uncovered easily from $H$ and $\widehat{H}$. This is because the decoding algorithm depends on some very specific graph theoretical properties of the complete graph, and the conjugates of $H$ cannot be modeled as the symmetric group acting on a complete graph. The parameters of this code mean that any cryptosystem using this is susceptible to ISD attacks. The question is, can we come up with a permutation code with parameters that make it resistant to ISD attacks and a decoding algorithm that cannot be modified for its conjugates? Our cryptosystem using such a permutation code would be secure against both kinds of attacks demonstrated earlier. 

	Permutation groups differ from vector spaces in their non-commutativity, and it would be an interesting question to see if this would lead to any significant improvements in key size, speed, etc over traditional code based cryptosystems. For example, a rank $k$ subspace of $\mathbb{F}_q^n$ needs $k$ basis vectors to describe, each of length $n$. On the contrary, very large permutation groups can be generated by a very small number of generators. The symmetric group itself for example, can be generated using just two of its elements. This is a characteristic shared by many non-abelian permutation groups. Hence, a cryptosystem using permutation codes can potentially achieve a quadratic reduction in key size compared to its linear code counterparts for the same level of security! A linear code based cryptosystem using a code of length $n$ over $\mathbb{F}_q$ would require a key of $O(n^2)$ as one of the components of the public key is a $k \times n$ matrix. By contrast, consider a permutation group of comparable size which is a subgroup of $S_n$ which can be generated using just two generators. The space needed to store these generators would be just $O(n)$. One of the common complaints against the McEliece cryptosystem is its large key size so this would be a direction of research worth pursuing.
\section*{Acknowledgements} This paper was part of the first author's Master's thesis at IISER Pune who was partially supported by an INSPIRE fellowship. The second author was partially supported by a NBHM research grant.	Both authors thank Upendra Kapshikar for stimulating discussions. We thank both the referees for their thorough reading of the paper and insightful comments.
\begin{small}
\renewcommand{\baselinestretch}{0.5}
\bibliography{thesis}
\end{small}
\Addresses
\end{document}